\documentclass[draft]{amsart}

\usepackage{amssymb,amsmath,tikz}

\newcommand{\Z}{\mathbb{Z}}
\newcommand{\C}{\mathbb{C}}
\newcommand{\R}{\mathbb{R}}

\newcommand{\imi}{\mathsf{i}}

\newtheorem{theorem}{Theorem}

\newtheorem{lemma}{Lemma}

\theoremstyle{remark}

\newtheorem{remark}{Remark}

\newtheorem{example}{Example}

\title{On beta pentagon relations}
\author{Rinat Kashaev}
\address{Section de math\'ematiques, Universit\'e de Gen\`eve,
2-4 rue du Li\`evre, 1211 Gen\`eve 4, Suisse\\}
\date{March 30, 2014}
\thanks{Supported in part by Swiss National Science Foundation}
\dedicatory{Dedicated to Professor Ludwig Faddeev on the occasion
     of his $80$th birthday}
\begin{document}

\begin{abstract}
The (quantum) pentagon relation underlies the existing constructions of three dimensional quantum topology in the combinatorial framework of triangulations.  Following the recent works \cite{KashaevLuoVartanov2012,AndersenKashaev2013}, we discuss a special type of integral pentagon relations and their relationships  with the Faddeev type operator pentagon relations.
\end{abstract}
\maketitle
\section{Introduction}
By \emph{pentagon relation} in a broad sense we mean any algebraic relation that can be given an interpretation in terms of Pachner's 2-3 move for triangulated three dimensional manifolds. A particular but still large class of pentagon relations are satisfied by $6j$-symbols arising from representation theory of Hopf algebras. A first historical example of this kind was published in 1953 by Biedenharn and Elliott \cite{MR0052585,Elliott1953} within the quantum theory of angular momentum.

Motivated by the problem of giving an exact combinatorial formulation for quantum Chern--Simons partition functions with non-compact gauge groups, the constructions of the works~\cite{KashaevLuoVartanov2012,AndersenKashaev2013} are based on five-term integral relations which  have a structure similar to that of the Biedenharn--Elliott relation but where the discrete variables, corresponding to the equivalence classes of irreducible representations of the group $SU(2)$, are replaced by continuous variables and the discrete sums by integrals. In this work, we put those relations into a general common framework and elucidate their connections with the Faddeev type operator pentagon relations. The framework is based on the theory of locally compact abelian groups and is defined as follows.

For two sets $S$ and $T$, we let $S^T$ denote the set of all maps from $T$ to $S$. For any non-negative integer $n$, we let $[n]$ denote the set $\Z_{\ge0}\cap\Z_{\le n}$ and $\Delta[n]$ denote  the standard combinatorial $n$-simplex seen, for example, as the power set $2^{[n]}$. We also denote by $\Delta[n]_i$ the set of $i$-dimensional simplexes of $\Delta[n]$.

Let $A$ be a locally compact abelian group together with a fixed Haar measure $dx$. We say that a function
\begin{equation}\label{eq-10}
\phi\colon [4]\times A^2\to\C,\quad (j,x,y)\mapsto \phi_j(x,y)
\end{equation}
is of the \emph{beta pentagon type} over $A$ if the following five term integral relation is satisfied:
\begin{equation}\label{eq00}
  \phi_1(x,y)\phi_3(u,v) =\int_{A}
 \phi_4(uy,v\bar z) 
 \phi_2(xyuv\bar z,z)  \phi_0(xv,y\bar z)dz
\end{equation}
where we denote $\bar x\equiv x^{-1}$. The relation~\eqref{eq00} itself will be called the \emph{beta pentagon relation} over $A$. Our motivation for it comes from the following combinatorial interpretation. 

Given a map as in \eqref{eq-10}, we define another map
\begin{equation}
W\colon \Delta[4]_3\times A^{\Delta[3]_1}\to \C
\end{equation}
by assigning
\begin{equation}
W(\partial_i\Delta[4],x)=\phi_i(x_{01}x_{23}\bar x_{03}\bar x_{12},x_{03}x_{12}\bar x_{02}\bar x_{13})
\end{equation}
where $x_{jk}$ is the $x$-image of the edge $\{j,k\}$. Then, the equality 
\begin{equation}
\prod_{i\in\{1,3\}}W(\partial_i\Delta[4],\varepsilon_i^{*}x)=
\int_{A}dx_{13}\prod_{i\in\{0,2,4\}}W(\partial_i\Delta[4],\varepsilon_i^{*}x),\quad \forall x\in A^{\Delta[4]_1},
\end{equation}
with the standard injections
\(
\varepsilon_i\colon [3]\to [4],\ i\in[4],
\)
defined by
\begin{equation}
\varepsilon_i(j)=\left\{
\begin{array}{cl}
 j&\mathrm{if}\ j<i;\\
 j+1&\mathrm{otherwise},
\end{array}
\right.
\end{equation}
is a consequence of the beta pentagon relation for  $\phi$.
Such a combinatorial interpretation can be used as a starting point for TQFT-like constructions\footnote{TQFT here stands for Topological Quantum Field Theory.} based on the combinatorics of triangulated pseudo 3-manifolds similarly to the Ponzano--Regge and Turaev--Viro models~\cite{PonzanoRegge1968,MR1191386,KashaevLuoVartanov2012,AndersenKashaev2013}.

An interesting example of a function of the beta pentagon type over $\R$ and the reason why we use the term ``beta" is given by the Euler beta function 
\begin{equation}
B(x,y)=\Gamma(x)\Gamma(y)/\Gamma(x+y).
\end{equation}
Namely, the solution is given explicitly by the formula~\cite{KashaevLuoVartanov2012}
\begin{equation}
\phi_j(x,y)=B\left(2\pi\imi(x+y-\imi 0),-2\pi\imi (y+\imi 0)\right).
\end{equation}
This solution has not been used for topological applications so far, because it lacks the tetrahedral symmetries, but it would be interesting to see if the pentagon relation could have any significance, for example, for the Veneziano amplitude in string theory.
\subsection*{Acknowledgements}  This work is supported in part by Swiss National Science Foundation. The results were reported at the workshop ``Extended TFT, quantization of classical TFTs, 
higher structures, differential cohomology theories, ..." held at the Erwin Schr\"odinger Institute, Vienna, February 17--21, 2014 and the conference 
``Mathematical Physics: Past, Present, and Future" held at the
Euler International Mathematical Institute, St. Petersburg, March 26--30, 2014. I would like to thank the participants of the workshop and the conference for useful and helpful discussions, especially A.~Alekseev, I.~Arefeva, D.~Bar-Natan, S.~Derkachov, L.~Faddeev, A.~Isaev, N.~Reshetikhin, R.~Schrader, M.~Semenov-Tian-Shansky, V.~Spiridonov, V.~Tarasov, L.~Takhtajan, J.~Teschner, A.~Virelizier.

\section{Some symmetries of the beta pentagon relation}
 One immediate symmetry of the beta pentagon relation is given by the inversion of the group arguments. Namely, if $\phi_i(x,y)$ is of the beta pentagon type over $A$, then
 \begin{equation}
\phi'_i(x,y)\equiv \phi_i(\bar x,\bar y)
\end{equation}
is also of the beta pentagon type over  $A$.

Let $\widehat A$ be the Pontryagin dual of $A$. Interestingly, the beta pentagon relation is stable under the Fourier transformation in the sense that if
$\phi_i(x,y)$ is of the beta pentagon type over $A$, then
\begin{equation}\label{eq:ft}
\hat\phi_i(\xi,\eta)\equiv \int_{A^2}\frac{\xi(y)}{\eta(x)}\phi_i(x,y)dxdy,\quad \forall(\xi,\eta)\in\widehat A^2,
\end{equation}
is of the beta pentagon type over $\widehat{A}$. This is checked as follows. 

First, for any integrable function $f\in\C^{A^2}$, we define its partial Fourier transform $\check f\in\C^{\hat A\times A}$ by the formula
\begin{equation}
\check f(\xi, y)\equiv \int_A \bar\xi(x)f(x,y)dx,\quad \forall(\xi,y)\in\widehat A\times A,
\end{equation}
so that we have
\begin{equation}\label{eq:check}
f(x,y)=\int_{\widehat A}\xi(x)\check f(\xi,y)d\xi
\end{equation}
and
\begin{equation}\label{eq:hatcheck}
\hat f(\eta,\xi)=\int_{A}\eta(y)\check f(\xi,y)dy.
\end{equation}
Now, we write
\begin{multline}
\hat\phi_1(\xi,\eta)\hat\phi_3(\mu,\nu) =\int_{A^4}\frac{\xi(y)\mu(v)}{\eta(x)\nu(u)}\phi_1(x,y)\phi_3(u,v)dxdydudv\\ 
=\int_{A^5}\frac{\xi(y)\mu(v)}{\eta(x)\nu(u)}\phi_4(uy,v\bar z) 
 \phi_2(xyuv\bar z,z)  \phi_0(xv,y\bar z)dzdxdydudv\\
=\int_{A^5}\frac{\xi(yz)\mu(vz)}{\eta(x\bar v\bar z)\nu(u\bar y\bar z)}\phi_4(u,v) 
 \phi_2(xu\bar z,z)  \phi_0(x,y)dzdxdydudv
\end{multline}
where in the last equality we have changed the variables $x\mapsto x\bar v$, $u\mapsto u\bar y$ followed by $y\mapsto yz$, $v\mapsto vz$. By applying \eqref{eq:check} to $\phi_2$, then collecting all characters applied to $z$, and applying  \eqref{eq:hatcheck}, we continue
\begin{multline*}
 =\int_{\widehat A\times A^5}\frac{\xi(yz)\mu(vz)\sigma(xu\bar z)}{\eta(x\bar v\bar z)\nu(u\bar y\bar z)}\phi_4(u,v) 
 \check\phi_2(\sigma,z)  \phi_0(x,y)d\sigma dzdxdydudv\\
=\int_{A^5\times\widehat A}\frac{\xi(y)\mu(v)\sigma(xu)\xi\eta\mu\nu\bar\sigma(z)}{\eta(x\bar v)\nu(u\bar y)}\phi_4(u,v) 
 \check\phi_2(\sigma,z)  \phi_0(x,y)dz  dxdydudvd\sigma\\
 =\int_{A^4\times\widehat A}\frac{\xi\nu(y)\mu\eta(v)}{\eta\bar\sigma(x)\nu\bar\sigma(u)}\phi_4(u,v) 
 \hat\phi_2(\xi\eta\mu\nu\bar\sigma,\sigma)  \phi_0(x,y) dxdydudvd\sigma
\end{multline*}
where in the last equality we have collected all characters applied to $x,y,u,v$. Finally, applying \eqref{eq:ft} to $\phi_0$ and $\phi_4$, we end up with the integral
\[
 =\int_{\widehat A}\hat\phi_4(\mu\eta,\nu\bar\sigma) 
 \hat\phi_2(\xi\eta\mu\nu\bar\sigma,\sigma)  \hat\phi_0(\xi\nu,\eta\bar\sigma)d\sigma.
\]
\section{Functions of the automorphic beta pentagon type}\label{sec:aug}
Let $A$ be a locally compact abelian group, $B\subset A$  a subgroup, $g\in\widehat{B}^{[2]}$, and $h\colon A\to \widehat B$, $x\mapsto h_x$, a group homomorphism such that the map $\varepsilon\colon B\to\mathbb{T}$, $b\mapsto h_b(b)$ is a character. We say that a function $\phi\in\C^{[4]\times A^2}$ is of the \emph{automorphic beta pentagon type}  $(B,g,h)$ if it satisfies the beta pentagon relation~\eqref{eq00} and the automorphicity conditions
\begin{equation}\label{eq:automorph}
\phi_i(bx,y)=\gamma_ih_y(b)\phi_i(x,y),\quad \forall (i,b,x,y)\in[4]\times B\times A^2,
\end{equation}
where
\begin{equation}\label{eq:gammai}
\gamma_0=g(0),\ \gamma_1=g(0)g(1),\ \gamma_2=g(1),\ \gamma_3=g(1)g(2),\ \gamma_4=g(2).
\end{equation}
Our first main result is the following theorem.
\begin{theorem}\label{thm1}
 Let $\phi\in\C^{[4]\times A^2}$ be of the automorphic beta pentagon type  $(B,g,h)$. Then, for any characters $\alpha,\beta\in\widehat{B}$, the function $\psi\in\C^{[4]\times A^2}$ defined by
 \begin{equation}
\psi_i(x,y)=\int_B\phi_i(x,yb)\mu_ih_{bx}(b)db,\quad\forall (i,x,y)\in[4]\times A^2,
\end{equation}
where
\begin{equation}\label{eq:mui}
\mu_0=\alpha\gamma_3,\ \mu_1=\alpha,\ \mu_2=\alpha\beta\gamma_0\gamma_2\gamma_4,\ \mu_3=\beta,\ \mu_4=\beta\gamma_1,
\end{equation}
satisfies the relations
\begin{multline}\label{eq:autopsi}
\psi_i(bx,y)=\gamma_ih_y(b)\psi_i(x,y),\\
\psi_i(x,by)=\varepsilon \bar \mu_i h_{\bar x}(b)\psi_i(x,y),\quad\forall (i,b,x,y)\in[4]\times B\times A^2,
\end{multline}
\begin{equation}\label{eq11}
  \psi_1(x,y)\psi_3(u,v) =\int_{A/B}
 \psi_4(uy,v\bar z) 
 \psi_2(xyuv\bar z,z)  \psi_0(xv,y\bar z)dz.
\end{equation}
\end{theorem}
Before proving the theorem, we start with few auxiliary facts.
\begin{lemma}
Given a locally compact abelian group $A$,   a subgroup $B\subset A$ , and $h\colon A\to \widehat B$, $x\mapsto h_x$, a group homomorphism. Then, the map $\varepsilon\colon B\to\mathbb{T}$, $b\mapsto h_b(b)$ is a group homomorphism, if and only if 
 \begin{equation}\label{eq:bbp}
h_b(c)h_{c}(b)=1,\quad \forall (b,c)\in B^2.
\end{equation}
\end{lemma} 
\begin{proof} We have 
\begin{equation}
 \varepsilon (bc)=h_{bc}(bc)=h_{bc}(b)h_{bc}(c)=h_{b}(b)h_{c}(b)h_{b}(c)h_{c}(c)=\varepsilon(b)\varepsilon(c)h_{c}(b)h_{b}(c)
\end{equation}
so that
\begin{equation}
 \varepsilon (bc)=\varepsilon(b)\varepsilon(c)\Leftrightarrow h_{c}(b)h_{b}(c)=1.
\end{equation}
\end{proof}
\begin{lemma}\label{lemma2}
 Given a locally compact abelian group $A$, a subgroup $B\subset A$, a character $\gamma\in\widehat{B}$, and a group homomorphism $h\colon A\to\widehat{B}$, $x\mapsto h_x$, such that \eqref{eq:bbp} is satisfied.
Let $f\in\C^{A^2}$ be such that
 \begin{equation}
f(bx,y)=\gamma h_y(b)f(x,y),\quad \forall (b,x,y)\in B\times A^2.
\end{equation}
Then, the function $\tilde f\in\C^{A^2\times\widehat B}$ defined by
\begin{equation}\label{eq:tildet}
\tilde f(x,y,\xi)=\int_B \xi(b)f(x,by)db
\end{equation}
satisfies the relations
\begin{multline}
\tilde f(bx,y,\xi)=\gamma h_{y}(b)\tilde f(x,y,\xi h_{\bar b}),\\
\tilde f(x,by,\xi)=\bar\xi(b)\tilde f(x,y,\xi),\quad \forall (b,x,y,\xi)\in B\times A^2\times\widehat B.
\end{multline}
\end{lemma}
\begin{proof} We have
 \begin{multline}
\tilde f(bx,y,\xi)=\int_B\xi(b')f(bx,b'y)db'=\int_B\xi(b')\gamma h_{b'y}(b)f(x,b'y)db'\\
=\gamma h_{y}(b)\int_B\xi(b') h_{b'}(b)f(x,b'y)db'=\gamma h_{y}(b)\int_B\xi(b') h_{\bar b}(b')f(x,b'y)db'\\
=\gamma h_{y}(b)\tilde f(x,y,\xi h_{\bar b})
\end{multline}
and
\begin{equation}
\tilde f(x,by,\xi)=\int_B\xi(b')f(x,b'by)db'=\int_B\xi(b'\bar b)f(x,b'y)db'=\bar\xi(b)f(x,y,\xi).
\end{equation}
\end{proof}
\begin{proof}[Proof of Theorem~\ref{thm1}]
We start by remarking that
\begin{equation}\label{eq:psiom}
\psi_i(x,y)=\omega_i(x,y,\mu_i),
\end{equation}
where
\begin{equation}
\omega_i(x,y,\xi)\equiv\tilde\phi_i(x,y,\varepsilon h_x\xi),\quad \forall (i,x,y,\xi)\in[4]\times A^2\times\widehat B,
\end{equation}
so that, by using Lemma~\ref{lemma2}, we verify \eqref{eq:autopsi}:
\begin{multline*}
\psi_i(bx,y)=\omega_i(bx,y,\mu_i)=\tilde\phi_i(bx,y,\varepsilon h_{bx}\mu_i)\\
=\gamma_ih_y(b)\tilde\phi_i(x,y,\varepsilon h_{x}\mu_i)
=\gamma_ih_y(b)\omega_i(x,y,\mu_i)=\gamma_ih_y(b)\psi_i(x,y)
\end{multline*}
and
\begin{multline*}
\psi_i(x,by)=\omega_i(x,by,\mu_i)=\tilde\phi_i(x,by,\varepsilon h_{x}\mu_i)\\
=\varepsilon \bar \mu_i h_{\bar x}(b)\tilde\phi_i(x,y,\varepsilon h_{x}\mu_i)
=\varepsilon \bar \mu_i h_{\bar x}(b)\omega_i(x,y,\mu_i)=\varepsilon \bar \mu_i h_{\bar x}(b)\psi_i(x,y).
\end{multline*}
To verify \eqref{eq11}, we proceed as follows:
\begin{multline*}\omega_1(x,y,\xi \varepsilon h_{\bar x})\omega_3(u,v,\eta \varepsilon h_{\bar u})\\
=
 \tilde\phi_1(x,y,\xi)\tilde\phi_3(u,v,\eta)=\int_{B^2}\xi(b)\eta(c)\phi_1(x,by)\phi_3(u,cv)dbdc\\
= \int_{A\times B^2}\xi(b)\eta(c)\phi_4(uby,cv\bar z) 
 \phi_2(xbyucv\bar z,z)  \phi_0(xcv,by\bar z)dzdbdc\\
  \end{multline*}
 where in the last equality we have used the beta pentagon relation for $\phi$. We continue by applying \eqref{eq:automorph} to $\phi_4$ and $\phi_0$
 \begin{multline*}
 = \int_{A\times B^2}\xi(b)\eta(c)\gamma_4h_{cv\bar z}(b)\gamma_0h_{by\bar z}(c)\phi_4(uy,cv\bar z) 
 \phi_2(bcxyuv\bar z,z)  \phi_0(xv,by\bar z)dzdbdc\\
 = \int_{A\times B^2}\xi\gamma_4h_{v\bar z}(b)\eta\gamma_0h_{y\bar z}(c)\phi_4(uy,cv\bar z) 
 \phi_2(bcxyuv\bar z,z)  \phi_0(xv,by\bar z)dzdbdc\\
 \end{multline*}
 where we have used \eqref{eq:bbp}. We continue by applying \eqref{eq:automorph} to $\phi_2$
  \begin{multline*}
  = \int_{A\times B^2}\xi\gamma_4h_{v\bar z}(b)\eta\gamma_0h_{y\bar z}(c)\gamma_2h_z(bc)\phi_4(uy,cv\bar z) 
 \phi_2(xyuv\bar z,z)  \phi_0(xv,by\bar z)dzdbdc\\
  = \int_{A\times B^2}\xi\gamma_4\gamma_2h_{v}(b)\eta\gamma_0\gamma_2h_{y}(c)\phi_4(uy,cv\bar z) 
 \phi_2(xyuv\bar z,z)  \phi_0(xv,by\bar z)dzdbdc\\
  = \int_{A}\tilde\phi_4(uy,v\bar z,\eta\gamma_1h_{y}) 
 \phi_2(xyuv\bar z,z)  \tilde\phi_0(xv,y\bar z,\xi\gamma_3h_{v})dz\\
 \end{multline*}
 where in the last equality we have used \eqref{eq:gammai} and the transformation~\eqref{eq:tildet}. We continue by splitting the integral over $A$ into a double integral over $B$ and  $A/B$:
  \begin{multline*}
 = \int_{B\times A/B}\tilde\phi_4(uy,v\bar z\bar b,\eta\gamma_1h_{y}) 
 \phi_2(xyuv\bar z\bar b,bz)  \tilde\phi_0(xv,y\bar z\bar b,\xi\gamma_3h_{v})dbdz\\
 = \int_{B\times A/B}\xi\eta\gamma_1\bar\gamma_2\gamma_3\varepsilon h_{yv\bar z}(b)\tilde\phi_4(uy,v\bar z,\eta\gamma_1h_{y}) 
 \phi_2(xyuv\bar z,bz)  \tilde\phi_0(xv,y\bar z,\xi\gamma_3h_{v})dbdz\\
  \end{multline*}
 where  we have used Lemma~\ref{lemma2} and \eqref{eq:automorph}. We end up the calculation by absorbing the integral over $B$ by using the definition~\eqref{eq:tildet}:
 \begin{multline*}
 = \int_{A/B}\tilde\phi_4(uy,v\bar z,\eta\gamma_1h_{y}) 
 \tilde\phi_2(xyuv\bar z,z,\xi\eta\gamma_1\bar\gamma_2\gamma_3\varepsilon h_{yv\bar z})  \tilde\phi_0(xv,y\bar z,\xi\gamma_3h_{v})dz\\
 = \int_{A/B}\omega_4(uy,v\bar z,\varepsilon\eta\gamma_1h_{\bar u}) 
 \omega_2(xyuv\bar z,z,\xi\eta\gamma_1\bar\gamma_2\gamma_3h_{\bar x\bar u})  \omega_0(xv,y\bar z,\varepsilon\xi\gamma_3h_{\bar x})dz.
\end{multline*}
The obtained equality, by substitutions $\xi=\alpha\varepsilon h_x$ and $\eta=\beta\varepsilon h_u$ takes the form
\begin{multline}
 \omega_1(x,y,\alpha)\omega_3(u,v,\beta)\\=
 \int_{A/B}\omega_4(uy,v\bar z,\beta\gamma_1) 
 \omega_2(xyuv\bar z,z,\alpha\beta\gamma_1\bar\gamma_2\gamma_3)  \omega_0(xv,y\bar z,\alpha\gamma_3)dz
\end{multline}
which, by using \eqref{eq:mui} and \eqref{eq:psiom}, is equivalent to \eqref{eq11}.
\end{proof}
\section{Functions of the Faddeev type} 
A function
\begin{equation}
f\colon[4]\times\R\to\C,\quad (i,x)\mapsto f_i(x),
\end{equation}
is called of \emph{the Faddeev type} if it satisfies the operator relation
\begin{equation}\label{eq35}
f_1(\mathsf{p})f_3(\mathsf{q})=f_4(\mathsf{q})f_2(\mathsf{p}+\mathsf{q})f_0(\mathsf{p})
\end{equation}
where $\mathsf{p}$ and $\mathsf{q}$ are self-adjoint operators in a Hilbert space satisfying Heisenberg's commutation relation
\begin{equation}
\mathsf{p}\mathsf{q}-\mathsf{q}\mathsf{p}=(2\pi\imi )^{-1}.
\end{equation}
\begin{lemma}
 A square integrable function $f$ is of the Faddeev type if and only if
 \begin{equation}\label{eq50}
\tilde f_1(x)\tilde f_3(y)=e^{-2\pi\imi xy}\int_\R \tilde f_4(y-z)\tilde f_2(z)\tilde f_0(x-z)e^{\pi\imi  z^2}dz,\quad\forall (x,y)\in\R^2,
\end{equation}
where
\begin{equation}
\tilde f(x)\equiv \int_\R e^{-2\pi\imi xy}f(y)dy.
\end{equation}
\end{lemma}
\begin{proof}
By using the inverse Fourier transformation, equality~\eqref{eq35} takes the form of an operator valued  integral equality
\begin{multline}
\int_{\R^2}\tilde f_1(x)\tilde f_3(y)e^{2\pi\imi  x\mathsf{p}}e^{2\pi\imi  y\mathsf{q}}dxdy\\=\int_{\R^3}\tilde f_4(y)\tilde f_2(z)\tilde f_0(x)e^{2\pi\imi  y\mathsf{q}}e^{2\pi\imi  z(\mathsf{p}+\mathsf{q})}e^{2\pi\imi  x\mathsf{p}}dxdydz
\end{multline}
which, by using the operator equalities
\begin{equation}
e^{2\pi\imi  x\mathsf{p}}e^{2\pi\imi  y\mathsf{q}}=e^{(2\pi\imi )^2 xy[\mathsf{p},\mathsf{q}]}e^{2\pi\imi  y\mathsf{q}}e^{2\pi\imi  x\mathsf{p}}
=e^{2\pi\imi xy}e^{2\pi\imi  y\mathsf{q}}e^{2\pi\imi  x\mathsf{p}}
\end{equation}
and
\begin{equation}
e^{2\pi\imi  z(\mathsf{p}+\mathsf{q})}=e^{\frac12(2\pi\imi )^2z^2[\mathsf{p},\mathsf{q}]}e^{2\pi\imi  z\mathsf{q}}e^{2\pi\imi  z\mathsf{p}}=e^{\pi\imi  z^2}
e^{2\pi\imi  z\mathsf{q}}e^{2\pi\imi  z\mathsf{p}},
\end{equation}
takes the form
\begin{multline}\label{eq40}
\int_{\R^2}\tilde f_1(x)\tilde f_3(y)e^{2\pi\imi xy}e^{2\pi\imi  y\mathsf{q}}e^{2\pi\imi  x\mathsf{p}}
dxdy\\=\int_{\R^3}\tilde f_4(y)\tilde f_2(z)\tilde f_0(x)e^{\pi\imi  z^2}e^{2\pi\imi  (y+z)\mathsf{q}}e^{2\pi\imi  (x+z)\mathsf{p}}dxdydz\\
=\int_{\R^3}\tilde f_4(y-z)\tilde f_2(z)\tilde f_0(x-z)e^{\pi\imi  z^2}e^{2\pi\imi  y\mathsf{q}}e^{2\pi\imi  x\mathsf{p}}dxdydz.
\end{multline}
Comparing the coefficients of the operators $e^{2\pi\imi  y\mathsf{q}}e^{2\pi\imi  x\mathsf{p}}$, we come to the conclusion that  equality~\eqref{eq40} is true if and only if equality~\eqref{eq50} is true.
\end{proof}
Taking the complex conjugate of \eqref{eq50}, we also have
\begin{equation}\label{eq60}
\bar{\tilde f}_1(x)\bar{\tilde f}_3(y)
=e^{2\pi\imi xy}\int_\R \bar{\tilde f}_4(y-z)\bar{\tilde f}_2(z)\bar{\tilde f}_0(x-z)e^{-\pi\imi  z^2}dz,\quad\forall (x,y)\in\R^2.
\end{equation}
\begin{remark}
 If $f_i(x)$ is a function of the Faddeev type, then the function $g_i(x)=f_i(-x)$ is also of the Faddeev type.
\end{remark}
\begin{example}
 The constant unit function 
 \begin{equation}
f_i(x)=1,\quad \forall i\in[4],
\end{equation}
is trivially of the Faddeev type.
\end{example}
\begin{example}
 The Gaussian exponentials
 \begin{equation}
f_j(x)=a_je^{-b_jx^2},\quad \forall j\in[4],\quad a\in\R^{[4]},\quad b\in\R^{[4]}_{>0},
\end{equation}
form a function of the Faddeev type with an appropriate choice of the constants $a_j$ and $b_j$.
\end{example}
\begin{example}
 A non-trivial and interesting example of a function of the Faddeev type is given by Faddeev's quantum dilogarithm \cite{MR1345554}:
 \begin{equation}\label{eq10}
f_j(x)=\Phi_\hbar(x)\equiv\exp\left(\int_{\R+\imi\epsilon}\frac{e^{-\imi 2xz}}{4\sinh(zb)\sinh(zb^{-1})}\frac{dz}z\right), \quad \forall j\in[4],
\end{equation}
where $\hbar\in\R_{>0}$, $b$ is any root of the equation 
\begin{equation}
(b+b^{-1})^{-2}=\hbar,
\end{equation}
and $\epsilon$ is an arbitrary sufficiently small positive real. It is convenient to choose the unique $b$ such that $0\le \arg b<\pi/2$, and $\sqrt{\hbar}=(b+b^{-1})^{-1}>0$.

The integral in \eqref{eq10} is absolutely convergent, and it admits analytic continuation to complex $x$ with $|\Im x|<\frac1{2\sqrt{\hbar}}$. In the case when $\arg b>0$ (i.e. $\hbar>\frac14$) one can show that
\begin{equation}\label{eq20}
\Phi_\hbar(x)=
\frac{(-qe^{2\pi b x};q^2)_\infty}{
(-\bar qe^{2\pi b^{-1}x};\bar q^2)_\infty},\quad
q:= e^{\pi\imi b^2},\ 
\bar q:= e^{-\pi\imi b^{-2}},
\end{equation}
where we use the standard notation of the theory of basic hypergeometric series
\begin{equation}
 (x;q)_\infty\equiv\prod_{n=0}^\infty(1-xq^n),\quad |q|<1.
\end{equation}
Equation~\eqref{eq20}, can be used to analytically continue the definition of $\Phi_\hbar(x)$ to the entire complex plane. It is straightforward to see that it satisfies the functional equations
\begin{equation}
\Phi_\hbar(x-\imi b^{\pm1}/2)=(1+e^{2\pi b^{\pm1}x })\Phi_\hbar(x+\imi b^{\pm1}/2).
\end{equation}
One has also the inversion relation
\begin{equation}\label{eq15}
\Phi_\hbar(x)\Phi_\hbar(-x)=e^{\pi \imi x^2}e^{-\pi\imi (2+\hbar^{-1})/12}.
\end{equation}
Faddeev's quantum dilogarithm is closely related to Shintani's double sine function \cite{MR0460283,MR1105522,Barnes1904}, but the pentagon identity~\eqref{eq35} seems not to be known before Faddeev's paper \cite{Faddeev1994}. For further properties of Faddeev's quantum dilogarithm see, for example, \cite{MR1828812,MR2171695,MR1770545,MR2952777}.
\end{example}
Our second main result is the following theorem.
\begin{theorem}\label{theorem}
 Let $f$ and $g$ be two functions of the Faddeev type. Then the function 
 \begin{equation}\label{eq70}
\varphi\colon [4]\times\R^2\to\C,\quad
(j,x,y)\mapsto \varphi_j(x,y)=\int_\R e^{2\pi\imi  yt}f_j\left(t+\frac x2\right)\bar{g}_j\left(t-\frac x2\right)dt
\end{equation}
is of the beta pentagon type over $\R$, i.e. it satisfies relation~\eqref{eq00} with $A=\R$.
\end{theorem}
It is instructive to give an operator interpretation for the formula in \eqref{eq70}. If we define the Fourier  operator $\mathsf{F}$ by
\begin{equation}
(\mathsf{F}f)(x)=\int_\R e^{2\pi\imi  xy}f(y)dy,
\end{equation}
which is unitary in $L^2(\R)$ due to the Fourier inversion formula,
we have the relations
\begin{equation}\label{eq72}
\mathsf{p}\mathsf{F}=\mathsf{F}\mathsf{q},\quad \mathsf{q}\mathsf{F}=-\mathsf{F}\mathsf{p},
\end{equation}
where
\begin{equation}
\mathsf{p}f(x)=\frac{1}{2\pi\imi }\frac{\partial f(x)}{\partial x},\quad \mathsf{q}f(x)=xf(x).
\end{equation}
By using Dirac's bra-ket notation for the scalar product in $L^2(\R)$:
\begin{equation}
\langle f\vert g\rangle\equiv\int_\R \bar f(x) g(x)dx,
\end{equation}
for any $(x,y)\in\R^2$, we have
\begin{multline}\label{eq73}
\varphi_j(x,y)\equiv\int_\R e^{2\pi\imi  yt}f_j\left(t+\frac x2\right)\bar g_j\left(t-\frac x2\right)dt\\
=\int_{\R}\overline{\left(e^{-\pi\imi x\mathsf{p}}g_j\right)}(t)e^{2\pi\imi  yt}\left( e^{\pi\imi x \mathsf{p}}f_j\right)(t)dt
=\int_{\R}\overline{\left(e^{-\pi\imi x\mathsf{p}}g_j\right)}(t)\left( e^{2\pi\imi  y\mathsf{q}}e^{\pi\imi x \mathsf{p}}f_j\right)(t)dt\\
=\langle e^{-\pi\imi x\mathsf{p}}g_j\vert e^{2\pi\imi y\mathsf{q}}e^{\pi\imi x\mathsf{p}}f_j\rangle
=\langle g_j\vert e^{\pi\imi x\mathsf{p}}e^{2\pi\imi y\mathsf{q}}e^{\pi\imi x\mathsf{p}}f_j\rangle
=\langle g_j\vert e^{2\pi\imi (x\mathsf{p}+y\mathsf{q})}f_j\rangle.
\end{multline}
This formula allows us to easily prove the following lemma.
\begin{lemma}
Let $\varphi_j(x,y)$ be defined as in \eqref{eq70}. Then the following equality holds true
 \begin{equation}\label{eq75}
\varphi_j(x,y)=\int_\R e^{2\pi\imi  xt}\tilde f_j\left(t-\frac y2\right)\bar{\tilde g}_j\left(t+\frac y2\right)dt,
\end{equation}
where 
\begin{equation}
\tilde f\equiv\mathsf{F}^{-1}f,\quad \forall f\in L^2(\R).
\end{equation}
\end{lemma}
\begin{proof}
Indeed, by using  \eqref{eq73} and \eqref{eq72}, we have
\begin{multline}
\varphi_j(x,y)=\langle g_j\vert e^{2\pi\imi (x\mathsf{p}+y\mathsf{q})}f_j\rangle=
\langle g_j\vert e^{2\pi\imi (x\mathsf{p}+y\mathsf{q})}\mathsf{F}\mathsf{F}^{-1}f_j\rangle\\=
\langle g_j\vert \mathsf{F}e^{2\pi\imi (x\mathsf{q}-y\mathsf{p})}\mathsf{F}^{-1}f_j\rangle=
\langle \mathsf{F}^{-1}g_j\vert e^{2\pi\imi (x\mathsf{q}-y\mathsf{p})}\mathsf{F}^{-1}f_j\rangle\equiv
\langle \tilde g_j\vert e^{2\pi\imi (x\mathsf{q}-y\mathsf{p})}\tilde f_j\rangle,
\end{multline}
and formula~\eqref{eq75} now follows by applying \eqref{eq73} backwards with $f_j$ and $g_j$ replaced with $\tilde f_j$ and $\tilde g_j$.
\end{proof}
\begin{proof}[Proof of Theorem~\ref{theorem}] By using \eqref{eq75}, we write
\begin{multline*}
 \varphi_1(x,y)\varphi_3(u,v)\\
 =\int_{\R^2}
 e^{2\pi\imi  (xs+ut)}\tilde f_1\left(s-\frac y2\right) \tilde f_3\left(t-\frac v2\right)\bar{\tilde g}_1\left(s+\frac y2\right)
\bar{\tilde g}_3\left(t+\frac v2\right)dsdt\\
\end{multline*}
and then continue by applying \eqref{eq50} and \eqref{eq60}
\begin{multline*}
 =\int_{\R^4}
 e^{2\pi\imi  \left(xs+ut+vs+yt\right)+\pi\imi (z^2-w^2)} \tilde f_4\left(t-\frac v2-z\right)\bar{\tilde g}_4\left(t+\frac v2-w\right)\\
 \times\tilde f_2(z) \bar{\tilde g}_2(w)
 \tilde f_0\left(s-\frac y2-z\right)
\bar{\tilde g}_0\left(s+\frac y2-w\right)
dsdtdzdw\\
\end{multline*}
changing the variables of integration $s\mapsto s+(z+w)/2$, $t\mapsto t+(z+w)/2$, the integrations over $s$ and $t$ can be performed by using \eqref{eq75}
\begin{multline*}
 =\int_{\R^2}
 e^{\pi\imi  \left(x+y+u+v+z-w\right)\left(z+w\right)}
 \varphi_4(y+u,v+z-w) \tilde f_2(z) \bar{\tilde g}_2(w)\varphi_0(x+v,y+z-w)
dzdw\\
 =\int_{\R^2}
 e^{\pi\imi  \left(x+y+u+v+z\right)\left(z+2w\right)}
 \varphi_4(y+u,v+z) \tilde f_2(z+w) \bar{\tilde g}_2(w)\varphi_0(x+v,y+z)
dzdw\\
\end{multline*}
where we have shifted $z\mapsto z+w$. Finally, by shifting $w\mapsto w-z/2$ and using again \eqref{eq75}, we obtain
\begin{multline*}
 =\int_{\R^2}
 e^{2\pi\imi  \left(x+y+u+v+z\right)w}
 \varphi_4(y+u,v+z) \tilde f_2\left(w+\frac z2\right) \bar{\tilde g}_2\left(w-\frac z2\right)\varphi_0(x+v,y+z)\\
\times dzdw
 =\int_{\R}
 \varphi_4(y+u,v+z) 
 \varphi_2(x+y+u+v+z,-z)  \varphi_0(x+v,y+z)
dz.
\end{multline*}
\end{proof}
\begin{example}\label{ex:kfv}
 If we take $f_j(x)=g_j(x)=\Phi_\hbar(x)$, then one calculates that
  \begin{multline}
\varphi_j(x,y)=\varphi^+(x,y)\equiv\int_{\R}\frac{\Phi_\hbar\left(t+\frac x 2\right)}{\Phi_\hbar\left(t-\frac x 2\right)}e^{2\pi \imi ty}dt\\
=\Psi_\hbar\left(x-\frac \imi{2\sqrt \hbar}\right)\Psi_\hbar\left(y+\frac \imi{2\sqrt \hbar}\right)\Psi_\hbar\left(-x-y+\frac \imi{2\sqrt \hbar}\right),
\end{multline}
where
\begin{equation}
\Psi_\hbar(x)\equiv\frac{\Phi_\hbar(x)}{\Phi_\hbar(0)}e^{-\pi \imi x^2/2}.
\end{equation}
The corresponding beta pentagon identity first has obtained  in \cite{KashaevLuoVartanov2012} as a limiting case of Spiridonov's elliptic beta-integral \cite{MR1846786}.
\end{example}
\begin{example}
 If we take $f_j(x)=g_j(-x)=\Phi_\hbar(x)$, then we have
  \begin{equation}
\varphi_j(x,y)=\varphi^-(x,y)\equiv\int_{\R}\frac{\Phi_\hbar\left(\frac x 2+t\right)}{\Phi_\hbar\left(\frac x 2-t\right)}e^{2\pi \imi ty}dt,
\end{equation}
and, unlike the previous case, it is not known, at least to the author, if the integral can be simplified any further. It is interesting to note that the function $\varphi^-(x,y)$ is real valued, and it is related to the previous example by Fourier transformation (cf. \eqref{eq:ft}), namely we have
\begin{equation}
2\varphi^-(-2x,-2y)=\int_{\R^2}e^{2\pi\imi(xv-yu)}\varphi^+(u,v)dudv.
\end{equation}
\end{example}
\begin{example}
  If we take $f_j(x)=\Phi_\hbar(x)$ and $g_j(x)=1$, we have
  \begin{equation}\label{eq:imp}
\varphi_j(x,y)=e^{-\pi\imi xy}(\mathsf{F}\Phi_\hbar)(y)=e^{-\pi\imi (x+y)y}\Phi_\hbar\left(y+\frac \imi{2\sqrt \hbar}\right)e^{\pi\imi(1+1/\hbar)/12}.
\end{equation}
This example has a specific quasi-periodicity property
\begin{equation}\label{eq:quasiper}
\varphi_j(x+1,y)=\varphi_j(x,y)e^{-\pi \imi y}.
\end{equation}
which allows to apply Theorem~\ref{thm1}, where $A=\R$, $B=\Z$, $g(i)=1$ for any $i\in [2]$ so that $\gamma_i=1$, and 
$h_x(m)=e^{-\pi\imi mx}$ so that the associated character is given by $\varepsilon(m)=h_m(m)=(-1)^m$.  By choosing $\alpha=\beta=1$, we also have $\mu_i=1$ and the corresponding function $\psi\in\C^{[4]\times \R^2}$ takes the form
 \begin{multline}
\psi_j(x,y)=\sum_{m\in\Z}\varphi_j(x,y+m)e^{-\pi\imi (x+m)m}\\=e^{-\pi\imi xy}\sum_{m\in\Z}(\mathsf{F}\Phi_\hbar)(y+m)e^{-2\pi\imi (x+1/2)m}\equiv\psi(x,y)
\end{multline}
with the quasi periodicity properties
\begin{multline}
\psi(x+m,y)=h_y(m)\psi(x,y)=e^{-\pi\imi my }\psi(x,y),\\
 \psi(x,y+m)=\varepsilon(m)h_x(-m)\psi(x,y)=
(-1)^me^{\pi\imi mx }\psi(x,y).
\end{multline}
and the five term integral relation
\begin{equation}
  \psi(x,y)\psi(u,v) =\int_0^1
 \psi(u+y,v-z) 
 \psi(x+y+u+v-z,z)  \psi(x+v,y-z)dz.
\end{equation}
Notice, that the integrand is a periodic function of $z$ as it should be according to the general theory of Section~\ref{sec:aug}. In \cite{AndersenKashaev2013}, it is shown that $\psi(x,y)$ admits an analytic continuation to $\C^2$ as a meromorphic function, and it has been used  for the reformulation of the Teichm\"uller TQFT.
\end{example}
 
\end{document}